\documentclass{article}
\usepackage{amsmath}
\usepackage{amsthm}
\usepackage{amsfonts}
\usepackage{amssymb}
\usepackage{graphicx}
\usepackage[section]{algorithm}
\usepackage{algorithmic}

\newtheorem{theorem}{Theorem}[section]
\newtheorem{lemma}{Lemma}[section]

\DeclareGraphicsExtensions{.eps}

\begin{document}

\author{Asghar Asgharian Sardroud \thanks{Email: asgharian@aut.ac.ir}
 \and  Alireza Bagheri \thanks{Email:ar\_bagheri@aut.ac.ir}\\
 Department of Computer Engineering \& IT,\\ Amirkabir University of Technology, Tehran, Iran.
 }
\title{An approximation algorithm for the longest cycle problem in solid grid graphs}


\maketitle
\abstract{
Although, the Hamiltonicity  of solid grid graphs are polynomial-time decidable,
the complexity of the longest cycle problem in these graphs is still open.
In this paper, by presenting a linear-time constant-factor approximation algorithm,
we show that the longest cycle problem in solid grid graphs is in APX.
More precisely, our algorithm finds a cycle of length at least $\frac{2n}{3}+1$ in 2-connected $n$-node solid grid graphs.

\textbf{Keywords:} Longest cycle, Hamiltonian cycle, Approximation algorithm, Solid grid graph.
}

\section{Introduction}\label{sec:intro}

The longest cycle and path problems are well-known NP-hard problems in graph theory.
There are various results which show that these problems are hard to approximate in general graphs.
For example, assuming that P$\ne$NP,
 it has been shown that there is no polynomial-time constant-factor
 approximation algorithm for the longest path problem and also
 it is not possible to find a path of length $n-n^{\epsilon}$
 in polynomial-time in Hamiltonian graphs \cite{karger1997approximating}.
The Color coding technique introduced by Alon et al. \cite{alon1995color}
is one of the first approximation algorithms for these problems
which can find paths and cycles of length $\log n$.
Later, Bj{\"o}rklund et al. introduced another technique with better approximation  ratio, i.e. O$(n\log \log n/\log^2 n)$,
for finding long paths \cite{bjorklund2003finding,gabow2008finding}.
To our knowledge, the result of Gabow \cite{gabow2007finding},
which can find a cycle or path of length $\exp( \Omega (\sqrt{\log l / \log \log l})))$
in graphs with the longest cycle of length $l$,
is the best polynomial-time approximation algorithm for finding the longest cycles.
The results also show that these problems are hard to approximate even in
bounded-degree and Hamiltonian graphs \cite{chen2006approximating,feder2005finding}.
These problems are even harder to approximate in the case of directed graphs as showed in \cite{bjorklund2004approximating}.
For more related results on approximation algorithms on general graphs see \cite{bansal2004further,gabow2008finding2,uehara2005efficient}.

There are few classes of graphs in which the longest path or the longest cycle problems are polynomial
\cite{bulterman2002computing,gutin1993finding,ioannidou2009longest,keshavarz2010longestingrid,mertzios2012simple,uehara2007computing}.
In the case of grid graphs, Itai et al. \cite{itai1982hamiltonian} showed that
the Hamiltonian path and cycle problems are NP-complete.
Grid graphs are vertex-induced subgraphs of the infinite integer  grid $G^\infty$.
Later, Umans et al. showed that the Hamiltonicity of solid grid graphs,
 i.e. the grid graphs in which each internal face has length four, is
 decidable in polynomial time \cite{umans1997hamiltonian}.
However, to our knowledge, there is no result on finding or approximating the longest cycle in this class of graphs,
but there is only a $\frac{5}{6}$-approximation algorithm
for finding the longest paths in grid graphs that have square-free cycle covers \cite{zhang2011approximating}.
In this paper, we introduce a linear-time constant-factor approximation algorithm for the longest cycle problem
in solid grid graphs.
Our algorithm first finds a vertex-disjoint cycle set containing at least $\frac{2n}{3}+1$ of the vertices of a given 2-connected, $n$-vertex solid grid graph
and then merge them into a single cycle.

We organized the paper as follows. In section \ref{sec:pre}, we present the terminology and some
preliminary concepts. Our algorithm for finding the cycle set of the desired length is given in section \ref{sec:cycleset},
and in section  \ref{sec:merge} we show that these cycles can be merged into a single cycle of the same size.
Finally, in section \ref{sec:concl} we conclude the paper.

\section{Preliminaries} \label{sec:pre}
In this section, we present the definitions which is used during the paper and
the necessary concepts about solid grid graphs.
Grid graphs are vertex-induced subgraphs of the infinite integer grid whose
vertices are the integer coordinated points of the plane and there is
an edge between any two vertices of unit distance.
Let $G$ be a solid grid graph, i.e. a grid graph that has no inner face of length more than four.
We consider solid grid graphs as plane graphs, considering their natural embedding
on the integer grid.
The vertices of $G$ adjacent to the outer face are called \emph{boundary vertices},
and the set of boundary vertices of $G$ form its \emph{boundary}.
The boundary of connected plane graph $G$ should be a closed walk,
i.e. a cycle in which vertices and edges may be repeated, which is called \emph{boundary walk} (considering that a single vertex is a walk of length zero).
We use $|W|$ to refer to the number of (not necessarily distinct) edges of a closed walk $W$.
Each \emph{cut vertex} of $G$, i.e. a vertex of $G$ that its removal makes $G$ disconnected, is a
 repeated vertex in its boundary walk and vice versa, therefor,
 $G$ is 2-connected, if and only if its boundary is a cycle.
If $G$ is not connected, its boundary should be a set of closed walks, i.e. the set of
 boundary walks of its connected components.
Let cycle $C$ be the boundary of $G$. We say a vertex of boundary cycle $C$ is \emph{convex vertex}, \emph{flat vertex} or \emph{concave vertex} respectively if
its degree in $G$ is two, three or four.
The embedding of any cycle of the plane graph $G$ is a simple rectilinear polygon, as a result,
in each cycle of $G$ the number of convex vertices should be four more than the number of concave vertices.
Also, note that, because solid grid graphs are vertex-induced, their boundary cycles
can not contain two consecutive concave vertices.
We define two edges of $G$ to be \emph{parallel edges} if they are not incident to a common vertex,
but both of them are adjacent to the same inner face.
When $G'$ is a subgraph of $G$, we use the notation $G\setminus G'$
to denote the graph obtained from $G$ after removing all the vertices of $G'$ and their incident edges.
It is easy to show that $G\setminus G'$ is also a solid grid graph,
 when $G'$ is the boundary of $G$, or it is a maximal 2-connected subgraph of $G$.
\section{Finding the Cycle Set}
\label{sec:cycleset}
Let $G$ be a 2-connected, $n$-node  solid grid graph and $C$ be its boundary cycle.
Given such a graph $G$, we present an algorithm that finds a set of vertex-disjoint cycles $S$ in $G$
containing at least $\frac{2n}{3}+1$ of the vertices of $G$.
In the next section, by merging these cycles, we construct a cycle of the desired length.

Let $S$ be initially empty.
We add $C$ to $S$, and since the $G\setminus C$ may be not 2-connected,
we repeat the process recursively on its 2-connected disjoint subgraphs.
Let $\{G'_1,...,G'_m\}$ be a maximal set of disjoint maximal 2-connected subgraphs of $G\setminus C$.
The pseudocode of the procedure for constructing cycle set $S$ is given in algorithm \ref{alg:1}.

\newcommand{\LET}{\STATE \textbf{let }}
\newcommand{\SET}{\STATE \textbf{set }}
\newcommand{\PROC}{\STATE \textbf{procedure }}
\newcommand{\FOREACH}{\STATE \textbf{foreach }}
\newcommand{\INDENT}{\hspace{.5cm}}

\begin{algorithm}
\caption{The algorithm of finding the cycle set $S$}
\label{alg:1}
\begin{algorithmic}
\PROC FindCycleSet($G$)
\end{algorithmic}
\begin{algorithmic}[1]
\STATE \textbf{if } $|G|<4$
    \STATE \INDENT \textbf{return }$\varnothing$
\LET $C$ be the boundary cycle of $G$
\LET $\{G'_1,...,G'_m\}$ be the set of subgraphs of $G$ as defined above
\STATE $S\leftarrow \{C\}$
\FOREACH {$G'_i,\ 1 \le i \le m$} \textbf{ do }
    \STATE \INDENT \textbf{if not} $G'_i$ is a connected component of $G\setminus C$ having $|G'_i|=4$
    \STATE \INDENT\INDENT $S\leftarrow S\ \cup $ FindCycleSet($G'_i$)
\RETURN $S$
\end{algorithmic}
\end{algorithm}

The line 7 of the algorithm, excludes some  length four cycles from $S$, because
these cycles may be unmergable in the next step of our algorithm.
The following lemma shows that the sum of the lengths of the cycles in $S$ is at least $\frac{2n}{3}+1$.
\begin{lemma}
\label{lem:sumcycle}
The sum of the length of the cycles in $S$ is at least $\frac{2n}{3}+1$.
\end{lemma}

To prove Lemma \ref{lem:sumcycle}, we need two other auxiliary lemmas,
so we defer it after the statements of Lemmas \ref{lem:innercyclesize} and \ref{lem:sumcycle2}.
Let $G_1,...,G_k$ be the connected components of $G\setminus C$ and $W_1,...,W_k$ be respectively their boundary walks.
Also, for each connected component $G_i$, let $C_i$ be the cycle of $G$ that immediately encloses $G_i$.
To be more precise, we can construct such cycles $C_1$, $C_2$, ... and $C_k$
by preforming a set of split operations on $C$ as follows.
For each pair of connected components $G_i$ and $G_j$,
their should be a pair of vertices $u$ and $v$ in $G$ whose removal disconnects
the vertices of $G_i$ from the vertices of $G_j$ in $G$.
Based on the fact that $u$ and $v$ are adjacent or not,
we use the split operation shown respectively in figure \ref{fig:split}(a) or (b)
to split the cycle $C$ into two cycles $C'$ and $C''$.
Note that, we need at most four new edges to construct $C'$ and $C''$ from $C$.
We repeat the split operation  recursively on $C'$ and $C''$ until
 we obtain $k$ cycles $C_1,..,C_k$ such that each cycle $C_i$, $1 \le i \le k$,
 encloses only a connected component $G_i$ of $G\setminus C$.
To construct the desired cycle set, only $k-1$ split operations are required,
and in each split operation we use at most four new edges. So, we should have the following equation:
\begin{equation}\label{equ:1}
 |C|+ 4(k-1)\ge \sum_{i=1}^{k}{|C_i|}
\end{equation}

Constructing these cycles is not necessary for finding the final long cycle, however, we need these cycles
in the proof of our lemmas.
In Lemma \ref{lem:innercyclesize}, we show that the length of the boundary walk of $G_i$ is less than the length of its enclosing cycle $C_i$.

\begin{figure}
\begin{center}
\vspace{1cm}
\includegraphics[scale=.6]{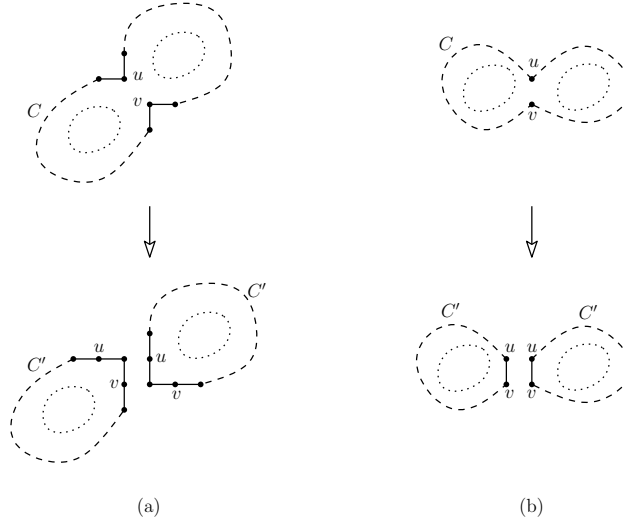}
\caption{The two split operations on $C$. The dashed lines schematically represent the cycle $C$ and the dotted lines schematically represent $G\setminus C$. }
\label{fig:split}
\end{center}
\end{figure}

\begin{lemma}
\label{lem:innercyclesize}
Let $G_i$ be a connected component of $G\setminus C$ and $W_i$ and $C_i$ be respectively its boundary walk and enclosing cycle.
Then we have $|C_i| \ge |W_i|+8$.
\end{lemma}
\begin{proof}
If $G_i$ be a single vertex then $|C_i|$ is at least eight and the lemma holds.
Otherwise, let $C_i$ and $W_i$ be directed in clockwise order.
For some edges of $W_i$ there is a distinct parallel edge in $C_i$ (as an example see the edges $e_1$ and $e_1'$ in  Figure \ref{fig:CiWi}).
Moreover, for each group of at most two consecutive edges of $W_i$ which have no parallel edges in $C_i$,
there is a distinct concave vertex in $C_i$
(for example the edges $e_2$ and $e_3$ and the vertex $u$ in Figure \ref{fig:CiWi}).
Instead, for each convex vertex $v$ of
$C_i$, the two edges of $C_i$ incident to $v$ have no parallel edge in $W_i$.
Therefore, knowing that the number of convex vertices in $C_i$ is equal  to the number of concave vertices plus four, we have $|W_i|\le|C_i|-8$.
\end{proof}


\begin{figure}
\begin{center}
\includegraphics[scale=.8]{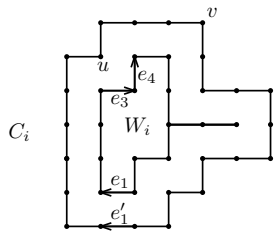}
\caption{An example cycle $C_i$ and walk $W_i$}
\label{fig:CiWi}
\end{center}
\end{figure}

Each 2-connected subgraph $G'_i$, $1 \le i \le m$, in Algorithm \ref{alg:1} is a subgraph of a connected component of $G\setminus C$.
Thus, without loss of generality, let $G'_{l_i},..,G'_{h_i}$ be subgraphs of $G_i$ for some $1\le l_i \le h_i \le m$.
Also, we call a vertex of $G$ \emph{free vertex} if it is not in any of cycles of $S$.
\begin{lemma}
\label{lem:sumcycle2}
For each connected component $G_i$ of $G\setminus C$, if $|G_i|>4$ we have $|G_i|- \sum_{j=l_i}^{h_i}|G'_j| \leq \frac{|W_i|}{2}+1$.
\end{lemma}
\begin{proof}
First, note that, because $G'_j$, $1\le j\le m$, is maximal and 2-connected,
 each vertex of $G_i$ which is not in any $G'_j$, $1\le j\le m$,
 should be on the boundary of $G_i$, i.e. $W_i$, and they should be free vertices.
Therefore, to prove the lemma, we show that at most $\frac{|W_i|}{2}+1$ of the vertices of $W_i$ are free vertices.
Let $W_i$ contains $x$ duplicated edges (i.e. the edges that are repeated in $W_i$ two times).
Removing all the $x$ duplicated edges from $W_i$, including the resulting  isolated vertices, will result a set of closed walks $Y_i$.
The length of $|Y_i|$, i.e. the sum of the lengths of its closed walks, should be $|W_i|-2x$.
Also, the vertices that are in $W_i$ but not in $Y_i$ should be free vertices, because they are adjacent only to the outer face.
There should be at most $x+1$ such distinct vertices.
In addition, we will show that there is at most $\frac{|Y_i|}{2}$  free vertices in $Y_i$.
Thus, the total number of free vertices of $W_i$ is not more than $\frac{|Y_i|}{2}+x+1$ which is equal to $\frac{|W_i|}{2}+1$.

Each free vertex of $Y_i$  should be adjacent to an inner face $f$ of $G_i$ which is not in any $G'_j$, $l_i \le j \le h_i$
(see Figure \ref{fig:free} as an example).
Because of the maximality of $G'_j$, $1\le j\le m$, $f$ can not share any edge with another inner face of $G_i$,
so it should be adjacent to a cut vertex $v_c$ of $G_i$.
Also, $v_c$ should be adjacent to another face $f'$ of $G_i$,
and because of maximality of $G'_j$, $1\le j\le m$, $f'$ can not contain free vertices.
Hence, $v_c$ is not a free vertex, and this ensures that $Y_i$ can not contain more than three consecutive free vertices.
Moreover, the fact that $f'$ can not contain free vertices shows that between any two group of consecutive free vertices in $Y_i$
 there is at least three consecutive non-free vertices.
Therefore, the number of free vertices in $Y_i$ is not more than $\frac{|Y_i|}{2}$. This completes our proof.
\end{proof}

\begin{figure}
\begin{center}
\includegraphics[scale=.7]{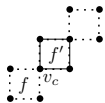}
\caption{An example in which three consecutive vertices of $Y_i$ are free vertices. Dotted edges are edges of $Y_i$ and dark edges are edges of the cycles in $S$}
\label{fig:free}
\end{center}
\end{figure}

\begin{proof}[Proof of lemma \ref{lem:sumcycle}]
To prove the lemma, we first show that in each iteration of construction of $S$ in algorithm \ref{alg:1},
the number of vertices left unused, i.e. the free vertices, is not more than $\frac{|C|-4}{2}$.
Consider a connected component $G_i$ of $G\setminus C$.
If $|G_i|\leq4$, all the vertices of $G_i$ is free vertices, and one can easily
check that $|G_i|< \frac{|C_i|-4}{2}$ holds.
Otherwise, using lemma \ref{lem:sumcycle2} the number of  vertices left free on $G_i$
is less than $\frac{|W_i|}{2}$ which is less than $\frac{|C_i|-8}{2}$ by lemma \ref{lem:innercyclesize}.
Next, summing the maximum number of vertices left free in each connected component
 of $G\setminus C$, the total number of vertices left free in a single step of
 the algorithm should be at most $\sum^{i=1}_{k}{\frac{|C_i|-4}{2}}$,
 which is not more than $\frac{|C|-4}{2}$ by Equation \ref{equ:1}.
Considering all the steps of the algorithm,
the total number of vertices left free can not be more than $\frac{|S|-4}{2}$.
Therefore, we have $n\leq |S|+ \frac{|S|-4}{2}$ which shows that $|S|\geq \frac{2n}{3}+1$.
\end{proof}

\section{Merging the Cycles}
\label{sec:merge}

After finding the cycle set $S$, the next step of the algorithm is to merge all the cycles of $S$ into a single cycle.
Except the boundary cycle of $G$, each cycle $C_{in}\in S$ is nested immediately inside a cycle $C_{out}\in S$ which is called its \emph{outer cycle}.
Also, $C_{in}$ is called an inner cycle of $C_{out}$.
Our algorithm starts from the outermost cycle of $S$, and merge
each cycle with its inner cycles using one of the merge operations which are shown in Figure \ref{fig:merge}.
But, $S$ may contains some cycles that are not mergeable with their outer cycles using our merge operations.
These cycles are diamond-shaped cycles. More precisely, a \emph{diamond-shaped cycle}
is a boundary cycle if it contains no flat vertex, and
a solid grid graph which its boundary is a diamond-shaped cycle is called a \emph{diamond-shaped grid graph}.

\begin{figure}
\begin{center}
\includegraphics[scale=.6]{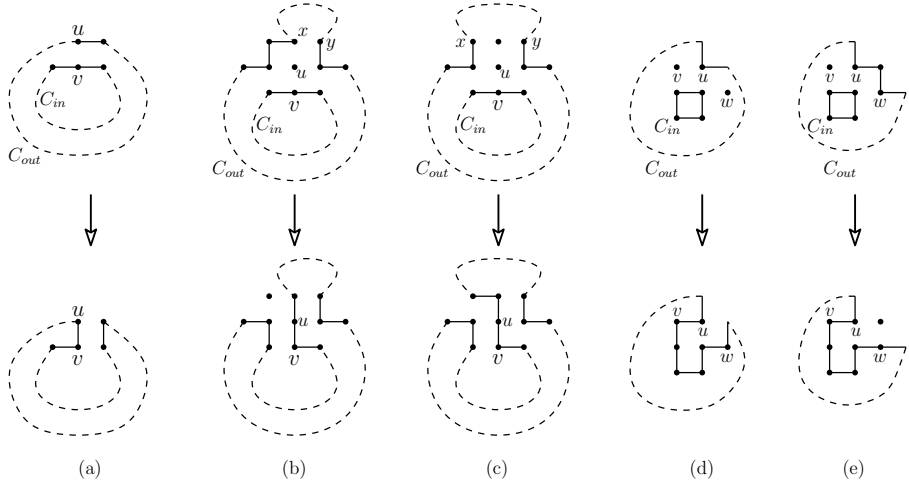}
\caption{Cycle merging operations}
\label{fig:merge}
\end{center}
\end{figure}

\begin{lemma} \label{lem:flatvertex}
Let $C_{in}$ be a cycle in $S$ and $C_{out}$ be its outer cycle.
For each flat vertex $v$ in $C_{in}$, there is at least one distinct flat vertex in $C_{out}$.
\end{lemma}
\begin{proof}
Let $G_{in}$ and $G_{out}$ be respectively the grid graphs that
$C_{in}$  and $C_{out}$  be their boundaries and $u$ be the vertex of $G$ outside of the cycle $C_{in}$ adjacent to $v$.
See the upper parts of Figures \ref{fig:merge} (a), (b) and (c).
If $u$ is not a free vertex, then it should be in $C_{out}$,
and at least one of its two incident edges in $C_{out}$ should be
parallel to one of the edges of $C_{in}$ incident to $v$ (upper part of Figure \ref{fig:merge} (a)).
But if $u$ is a free vertex,
 considering this fact that
 $C_{in}$ and $C_{out}$ are boundary cycles of some solid grid graphs
and $G_{in}$ is a maximal 2-connected subgraph of $G_{out}\setminus C_{out}$,
the configuration of $C_{in}$ and $C_{out}$  around the vertices $u$ and $v$
must be isomorphic to one of the configurations which is depicted in the upper parts of Figures \ref{fig:merge}(b) and (c).
Hence,
 in Figure \ref{fig:merge}(a) either $u$ or one of its two adjacent vertices in $C_{out}$,
 in Figure \ref{fig:merge}(b) one of the vertices $x$ and $y$,
 and in Figure \ref{fig:merge}(c) both of the vertices $x$ and $y$ are flat vertices of $C_{out}$.
Thus,
 for each flat vertex of $C_{in}$ there is a flat vertex in $C_{out}$.
\end{proof}

Let $C_i$ be a diamond-shaped cycle and $G_i$ be a solid grid graph which its boundary is $C_i$.
If $G_i\setminus C_i$ be 2-connected, it should be diamond-shaped.
Otherwise, by Lemma \ref{lem:flatvertex}, $C_i$ should has a flat vertex.
So, the diamond-shaped cycles in $S$ can be grouped into some groups of nested
diamond-shaped cycles (for an example see figure \ref{fig:convertdiamond}).
We have the following lemma about the innermost diamond-shaped cycles.
Note that, the length-four cycle is the smallest diamond-shaped cycle.

\begin{figure}
\begin{center}
\includegraphics[scale=.7]{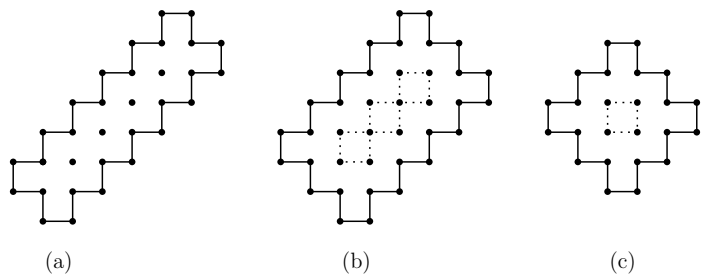}
\caption{Examples of diamond-shaped cycles. Only (a) and (c) can be innermost cycles in $S$}
\label{fig:diamond}
\end{center}
\end{figure}

\begin{lemma} \label{lem:freevertex}
Let $C_d$ be an innermost diamond-shaped cycle in $S$ and $G_d$ be the solid grid graph whose boundary is $C_d$.
Then, either $|C_d|=4$ and the outer cycle of $C_d$ is not diamond-shaped or there is at least one free vertex in the boundary of $G_d\setminus C_d$.
\end{lemma}
\begin{proof}
First let $|C_d|=4$, and  the outer cycle $C_{out}$ of $C_d$ is diamond shaped, as depicted in \ref{fig:diamond}(c),
and $G_{out}$ be the  solid grid graph which its boundary is $C_{out}$.
In this case, $G_d$ is a connected component of $G_{out}$ which contradicts the
line 7 of Algorithm \ref{alg:1}. Therefor, $C_{out}$ is not diamond-shaped when $|C_d|=4$.
For the case that $|C_d|>4$, because $C_d$ is an innermost cycle in $S$,
either $G_d\setminus C_d$ is a set of isolated vertices, as shown in Figure \ref{fig:diamond}(a),
or it is a length-four cycle not in $S$, as depicted in Figure \ref{fig:diamond}(b).
Clearly, in both cases there should be a free vertex in the boundary of $G_d\setminus C_d$.
\end{proof}

Using the free vertices that their existence proved in Lemma \ref{lem:freevertex},
we replace each group of nested diamond-shaped cycles in $S$,
except the length-four diamond-shaped cycles, by a set of non-diamond-shaped cycles as depicted in Figure \ref{fig:convertdiamond},
and we name the resulting cycle set $S'$.
So, $S'$  does not contain any diamond-shaped cycles, except the diamond-shaped cycles of length four.
Lemma \ref{lem:mergeall3} insures that we can merge all the cycles of $S'$ starting from the innermost cycles.

\begin{figure}
\begin{center}
\includegraphics[scale=.5]{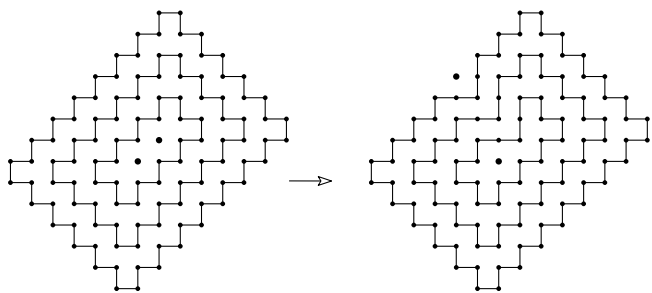}
\caption{Converting a group of nested diamond-shaped cycles into a group of non-diamond-shaped cycles using a free vertex }
\label{fig:convertdiamond}
\end{center}
\end{figure}

\begin{lemma} \label{lem:mergeall3}
Using the merge operations of Figure \ref{fig:merge}, all the cycles of $S'$ can be merged into a single cycle $C$ containing a given boundary edge $e$ of $G$.
\end{lemma}
\begin{proof}
If $S'$ contains only one cycle, the lemma holds easily. Otherwise,
let $C_{in}$ be a cycle of $S'$ and $C_{out}$ be its outer cycle,  and let
$G_{in}$ and $G_{out}$ be respectively the grid graphs that $C_{in}$ and $C_{out}$ are their boundary cycles.

For the case that $C_{in}$ is a non-diamond-shaped cycle,
let $v$ be a flat vertex of $C_{in}$.
As described in the proof of Lemma \ref{lem:flatvertex},
there is only three possible configurations for $C_{in}$ and $C_{out}$ around $v$,
as depicted in the upper parts of the Figures \ref{fig:merge} (a), (b) and (c).
In these cases, we can use respectively the merge operations depicted in
Figures \ref{fig:merge} (a), (b) and (c) to merge $C_{in}$ and $C_{out}$.
Moreover,
 $C_{in}$ should contain at least two flat vertices, because any cycle in a grid graph has even length and the
 number of convex vertices in $C_{in}$ is four more than concave vertices.
Therefore,
 starting from the outermost cycle and each time merging the cycle by one of its inner cycles
 one can merge all the non-diamond shaped cycles of $S'$ into a single cycle $C_{res}$.
Note that,
 the outermost cycle of $S'$ contains the edge $e$, and it can be merged by each of
 its inner cycles using at least two different flat vertices.
Therefor,
 we can chose the merge operations such that the cycle $C_{res}$ contains the edge $e$.

Considering Lemma \ref{lem:freevertex},
the only case that $C_{in}\in S'$ is diamond-shaped, is when $|C_{in}|=4$.
Let $C_{in}$ be such a cycle and let it has no parallel edge with $C_{res}$,
otherwise they can be merged by the merge operation of Figure \ref{fig:merge}(a).
In this case, by line 7 of Algorithm \ref{alg:1} and maximality of subgraphs $G'_i$,
there should be a free vertex $v$ adjacent to one of the four vertices of $C_{in}$.
Because there is no parallel edges between $C_{in}$ and $C_{res}$,
there are two possible configurations for $C_{in}$ and $C_{res}$ around the vertex $v$.
The two possible configurations are depicted in the upper parts of the Figure \ref{fig:merge} (d) and (e).
Thus, based on the fact that  vertex $w$,
in these figures, is a free vertex or not,
$C_{in}$ and $C_{res}$ can be merged respectively using  the merge
operations of the Figures \ref{fig:merge}(d) and (e).
\end{proof}

We conclude this section summarizing our result in the following theorem.

\begin{theorem} \label{thm:longestcycle}
There is a linear-time $\frac{2}{3}$-approximation algorithm for finding
a longest cycle in solid grid graphs.
\end{theorem}
\begin{proof}
The desired approximation algorithm is as follows.
Let $G$ be a largest 2-connected subgraph of a given solid grid graph.
First, construct the cycles set $S$ using Algorithm \ref{alg:1}, then
convert its diamond-shaped cycles to non-diamond shape cycles and make $S'$ as described before.
Constructive proof of Lemma \ref{lem:mergeall3} gives a method for merging all the cycles of $S'$,
and Lemma \ref{lem:sumcycle} ensure that the constructed long cycle contains at least two third
of the vertices of $G$.
We complete our proof arguing that the introduced approximation algorithm
can be implemented in linear time.
The boundary cycle $C$ of $G$ can be found in time $|C|$,
and by only checking the boundary vertices of $G$,
one can construct  a maximal set of disjoint 2-connected components of $G\setminus C$.
Thus, the lines 3 and 4 of Algorithm \ref{alg:1} can be implemented in time O$(|C|)$.
Moreover, except the recursive calls, the other lines of the algorithm
can be implemented in constant time.
Therefore, Algorithm \ref{alg:1} can be implemented such that
$FindCycleSet(G)$ runs in time O$(|S|)$.
The other steps of our algorithm, i.e.
constructing cycle set $S'$ from $S$,
finding the flat vertices of cycles of $S'$ and
merging the cycles of $S'$
can be implemented in linear time.
Thus, the total running time of the algorithm is O($|S|$)
which is linear with respect to the size of $G$.

\end{proof}

\section{Conclusions}
\label{sec:concl}
We introduced a linear-time approximation algorithm that, given a 2-connected, $n$-node solid grid graph,
can find a cycle containing at least two third of its vertices.
Since, cycles are 2-connected, our algorithm is a constant-factor approximation for
the longest cycle problem in solid grid graphs.
In other words, if the given solid grid graph $G$ is not 2-connected,
one can apply our algorithm to the largest 2-connected subgraph of $G$ to
find a cycle of the length at least two third of the length of the longest cycle of $G$.
A better approximation ratio for the longest cycle problem in solid grid graphs
or the longest path problem in this class of graphs can be the subject of future work.
\bibliographystyle{plain}
\bibliography{mybib}
\end{document}